\documentclass[11pt]{article}%
\usepackage{hyperref}
\usepackage{amsmath}
\usepackage{amstext} 
\usepackage{graphics}
\usepackage{color}
\usepackage{epsfig}
\usepackage{graphicx}%
\usepackage{amsfonts}%
\usepackage{amssymb}
\usepackage{setspace}
\usepackage[margin=1in]{geometry}
\usepackage{comment}
\usepackage{listings}
\usepackage{color}
\usepackage{hyperref} 
\usepackage[full]{optional}   

\newtheorem{theorem}{Theorem}[section]

\newtheorem{definition}{Definition}[section]
\newtheorem{claim}{Claim}[section]
\newtheorem{lemma}{Lemma}[section]

\newtheorem{corollary}{Corollary}[section]

\newtheorem{observation}{Observation}[section]

\newtheorem{proposition}{Proposition}[section]

\newcommand{\qed}{\hfill \mbox{\raggedright \rule{2mm}{3mm}}}

\newenvironment{proof}{\noindent{\bf Proof.}}{\qed}

\newcommand{\cfl}{{\sc Cfl}}

\newcommand{\opname}{\operatorname}
\newcommand{\opn}{\operatorname}
\newcommand{\zeone}{$0$-$1$ }
\newcommand{\noi}{\noindent} 
\date{}

\allowdisplaybreaks

\begin{document}

\title{
Extended Formulation Lower Bounds via Hypergraph Coloring?\thanks{
This research has been co-financed by the European Union (European
Social Fund -- ESF) and Greek national funds through the Operational
Program ``Education and Lifelong Learning'' of the National Strategic
Reference Framework (NSRF) - Research Funding Program:
``Thalis. Investing in knowledge society through the European Social Fund''.}
}

\author{Stavros G. Kolliopoulos\thanks{Department of Informatics and
Telecommunications, National and Kapodistrian 
University of Athens, Panepistimiopolis Ilissia, Athens
157 84, Greece; (\texttt{sgk@di.uoa.gr}).}   
\and Yannis Moysoglou\thanks{ 
Department of Informatics and
Telecommunications, National and Kapodistrian 
University of Athens, Panepistimiopolis Ilissia, Athens
157 84, Greece; (\texttt{gmoys@di.uoa.gr}). } }
\maketitle

\thispagestyle{empty}
\maketitle

\begin{abstract}
Exploring   the  power   of  linear   programming   for  combinatorial
optimization problems has been recently receiving renewed attention 
after a series of breakthrough impossibility results. 
From an algorithmic perspective, the
related questions concern whether there are compact formulations 
even for problems that are known to admit polynomial-time algorithms.

We propose a  framework for proving lower bounds on the
size of extended formulations. We do so by introducing a specific type
of extended relaxations that we call {\em product relaxations}  and is motivated
by the study of the Sherali-Adams (SA) hierarchy.  Then we show that for
every  approximate   relaxation  of  a   polytope  $P,$  there   is  a
product relaxation that  has the same size and  is at least as
strong. We provide a methodology for proving lower bounds on the size of approximate 
product relaxations by lower bounding the chromatic number of an underlying hypergraph,
whose vertices correspond to gap-inducing vectors.

We  extend  the  definition  of  product  relaxations  and  our
methodology to mixed integer sets. However in this  case we are able
to show that {\em mixed product relaxations} are at least as powerful
as  a special family  of extended  formulations. 
As  an application  of our
method we show an exponential lower bound on the size of approximate
mixed product relaxations for the metric capacitated facility
location problem (\cfl), a problem which seems to be intractable for linear programming as far as constant-gap compact formulations are concerned.
Our lower bound 
implies an unbounded integrality gap for \cfl\ at $\Theta({N})$
levels of the
{\sl universal} SA hierarchy which 
 is {\sl independent of the
starting relaxation;} we only require that the starting
relaxation has size $2^{o(N)}$, where $N$ is the number of facilities
in  the instance. This proof yields the  first such tradeoff 
for an SA procedure  that is independent of 
the initial relaxation.
\end{abstract}

\section{Introduction}

In  the  past few  years  there has  been  an  increasing interest  in
exposing the limitations of  compact LP formulations
for combinatorial optimization problems. The goal  is
to  show    a lower  bound   on  the  size  of  {\em extended
formulations (EFs)} for a particular  problem.  
Extended formulations add extra variables to the natural problem space;
the increase in dimension may yield a smaller number of facets. 
The minimum size over all extended formulations is the {\em extension complexity} of
the corresponding polytope. 
A superpolynomial lower bound on the extension complexity
is  of intrinsic interest
in polyhedral combinatorics and   implies  that there is no 
polynomial-time  algorithm relying  purely  on the  solution  of a  compact
linear  program. It  does not 
however rule out efficient LP-based algorithms that combine
algorithmic steps of arbitrary type, such as
preprocessing, primal-dual, etc., with linear programming.

In the seminal paper  of Yannakakis \cite{Yannakakis91} the problem of
lower bounding the size of  extended formulations  was considered for
the  first time: exponential  lower bounds  were proved  for symmetric
extended formulations of the matching and TSP polytopes.  Yannakakis 
\cite{Yannakakis91} identified also a crucial combinatorial parameter, the
nonnegative rank of the slack matrix of the underlying polytope $P,$ and he
showed that it equals the extension complexity of $P.$ 
A strong connection of the extension complexity of a
polytope to communication  complexity was made in \cite{Yannakakis91},
by 
showing that the nonnegative rank of the slack matrix is at least 
the size of its 
minimum rectangle cover. That connection has been 
exploited   in several  results on
the extension complexity of polytopes.

Fiorini  et al.  \cite{FioriniMPTD12}  lifted the  symmetry condition  on
the result of \cite{Yannakakis91} regarding the TSP polytope,
thus answering a  long-standing open problem. The  result was obtained
by  showing that  the correlation  polytope has  exponential extension
complexity which in  turn was shown using communication complexity
tools. Recently, Rothvo{\ss} \cite{Rothvoss14} removed the symmetry condition
for the matching  polytope as well, answering  the second
long-standing open question of \cite{Yannakakis91}. This was done by a
breakthrough in  bounding a refined version of  the rectangle covering
number.

A more  general question is that  of the size  of approximate extended
formulations. This  problem was first  considered in \cite{BraunFPS12}
where  the   methodology  of  \cite{FioriniMPTD12}   was  extended  to
approximate  formulations  and an  exponential  bound  for the  linear
encoding of the  $n^{1/2-\varepsilon}$-approximate clique problem was
given.  Subsequently,  Braverman   and  Moitra  \cite{BravermanM13}
extended   the  former   bound   to  $n^{1-\varepsilon}$--approximate
formulations  of the  clique,  following a  new, information  theoretic,
approach. Braun  and Pokutta  in \cite{BraunP13} further strengthened the lower
bounds by  introducing the notion of common  information. Very recently,
Braun   and   Pokutta   \cite{BraunP14}   extended   the   result   of
\cite{Rothvoss14} to approximate  formulations of the matching polytope
by combining ideas of the latter with the notion of common information.

 In \cite{ChanLRS13} it was 
proved that in terms of approximating
  maximum constraint satisfaction problems, LPs of size $O(n^k)$ are
exactly as powerful as $O(k)$-level  relaxations in the  Sherali-Adams
hierarchy. Their proof
differs from previous work  in showing that polynomials of low degree
can 
approximate the functional version of the factorization 
theorem of \cite{Yannakakis91}.

The  {\em metric capacitated  facility location}  problem (\cfl)  is a
well-studied problem for which, while constant-factor approximations 
are known \cite{BansalGG12,AggarwalLBGGGJ13}, 
no efficient LP relaxation with constant integrality gap is
known. An instance $I$ of \cfl\ is defined as follows.  We are given a set
$F$ of facilities and a set $C$  of clients, with each facility $i$ having a capacity $U_i$ and each client $j$ having a demand $d_j>0$. We may open facility $i$ by
paying its opening cost $f_i$ and we may assign demand from client $j$ to facility
$i$ by  paying the connection  cost $c_{ij}$ per unit of demand. The 
latter costs satisfy the following variant of the triangle inequality:
$c_{ij} \leq c_{ij'} + c_{i'j'} + c_{i'j}$ for any $i, i'\in F$ and
$j, j' \in C.$ We  are asked to  open a subset $F' \subseteq F$ of the facilities 
and assign all the client demand to the open  facilities while 
respecting the capacities.  The  goal  is  to  
minimize  the  total  opening  and connection cost. 
The question whether  an efficient relaxation exists for \cfl\ is among the most
important     open     problems     in    approximation     algorithms
\cite{ShmoysWbook}. In a recent breakthrough,  the first $O(1)$-factor 
LP-based algorithm for \cfl\ 
was given in \cite{DBLP:conf/focs/AnSS14}. 
The proposed relaxation is, however,  exponential in size and,
according to the authors of \cite{DBLP:conf/focs/AnSS14}, not known to be separable in polynomial
time. To our knowledge, there has not been a  compact EF  for approximate 
\cfl\ achieving even an 
$o(|F|)$ gap.

In previous work 
\cite{KolliopoulosM13,KolliopoulosM14}, we proved
among other results for \cfl, unbounded integrality gaps for the Sherali-Adams
hierarchy when starting from the natural LP relaxation. We also
disqualified, with respect to obtaining a  $O(1)$ gap, 
 valid inequalities from  the literature such as the
flow-cover inequalities and their generalizations.

\subsection{Our contribution}

In this  paper we propose a  new approach for proving  lower bounds on
the  size of  approximate extended  formulations. Our contribution is
summarized by the following.

First  we  introduce a  family  of  extended  relaxations of  a  given
polytope   which  we   call  \emph{product   relaxations}.  The
product  relaxations   are  inspired   by  the  study   of  the
Sherali-Adams hierarchy. 
Given a polytope $K \subseteq [0,1]^d$ 
that corresponds to a linear relaxation of the
problem at hand, the Sherali-Adams relaxation $\opn{SA}^t(K)$ 
at level $t$ is produced by a lift-and-project
method, where initially every constraint in the description of $K$ is
multiplied with all $t$-subsets of variables and their complements. 
The resulting products of variables are then
linearized, i.e., each replaced by  a single variable, and finally one
projects back to the original variable space of $K.$ 
The variable space of the product relaxations is exactly the space of the 
final $d$-level 
Sherali-Adams relaxation, after linearization and 
before projection.  The
variables have the intuitive meaning of corresponding to products over sets of
variables from the original space  -- the "intuitive meaning" of
a variable is made precise through the notion of the \emph{section}
$f$  of an extended  relaxation $Q(x,y)$ of a polytope $P(x)$, a
function $f$ that maps an integer point $x \in P(x)$ to a vector  of
values $(x,y)=f(x)$ 
such that $f(x) \in Q(x,y)$. (See Section~\ref{sec:prelim} for the
necessary definitions). 

We prove in Theorem~\ref{theorem:ext1} 
that for any $\rho$-approximate
extended  formulation   of  a  \zeone  polytope  there   is  a  product
relaxation of the same size that is at least as strong. The proof is
short and accessible. 
 Theorem~\ref{theorem:ext1}  reduces
  lower bounding  the  size of  an extended formulation, which  uses
  some unknown space and encoding,   of a
polytope $P,$  to  lower bounding
the size of  product relaxations of $P.$  In the 
product space we have the concrete  advantage of
  knowing the section of the target relaxation. 
We  extend  the  definition  of  product  relaxations  and  our
methodology to mixed integer sets. However in this  case we are able
to show that {\em mixed product relaxations} are at least as powerful
as  a special family  of extended  formulations 
(cf. Theorem~\ref{theorem:ext2}).

We note that our approach does not rely on the notion of 
the slack matrix introduced by Yannakakis \cite{Yannakakis91}. It 
differs  from that of \cite{ChanLRS13} in which 
the slack functions of the factorization theorem \cite{Yannakakis91}
were shown to be approximable, for max CSPs, by 
low-degree polynomials and thus   SA gaps   
are transferred to general  linear programs. 

Then we use a methodology for proving lower bounds for relaxations
for  which  the section  is  known  and in  particular  for
product relaxations.  Similar arguments have been used in the context
of bounding the number of facets of specific polyhedra, but prior to our work,
they seemed inapplicable for lower bounding the size of arbitrary EFs
which lift the polytope in   arbitrary variable spaces. 
The method  is the following: first define a set
of vectors in the space of the relaxation such that for each one of those vectors there is an admissible objective function witnessing an integrality gap of $\rho$.  We call that set of vectors the {\em core.} Then show
that, for any partition of the core into fewer than $\kappa$ parts,
there must be some part containing a set of conflicting vectors. A set of
infeasible vectors is {\em conflicting} if its convex hull has nonempty intersection
with the convex hull of $\{ f(x) \mid x \in P(x)\cap \{0,1\}^n \}$,
 which is always   included  in   the  feasible   region  of   a  product
relaxation -- here $f(x)$ is the section we associate with 
 product relaxations. Thus, we  get that at least $\kappa$  inequalities are needed
to separate  the members of  the core from  the feasible region  and so
$\kappa$ is a lower bound on the size of any $\rho$-approximate
product relaxation. By considering the hypergraph whose set of
vertices corresponds  to the aforementioned  set of vectors  and whose
set of hyperedges corresponds to  the sets of conflicting vectors, 
the chromatic number of the hypergraph is a lower bound on the size of
every $\rho$-approximate extended formulation
(cf. Theorem~\ref{theorem:hyper}).
Moreover, there is always a core such that the chromatic number of the
resulting, possibly infinite, 
hypergraph equals the extension complexity of the polytope
at hand. Thus the characterization of extension complexity in 
Theorem~\ref{theorem:hyper}
can be seen as an alternative to the nonnegative rank of the slack matrix.
The conflicting vectors are fractional solutions, which are hard to
separate from the integer solutions. The method comes closer to
standard LP/SDP integrality gap arguments than the existing 
combinatorial approaches for lower bounding extension complexity.

When arguing about  the polyhedral complexity of a specific polytope, 
i.e., the minimum size of its formulation in the original variable space, 
the above method can always
be  simplified to  finding a  set of  gap--inducing vectors  with the
property that (almost) any pair  of them are conflicting. The underlying
hypergraph reduces then  to a simple graph that  is very dense, almost
a clique, and thus has high chromatic number. 
We used this idea in a preliminary version of this work \cite{KolliopoulosM14b} 
to derive  exponential bounds on the polyhedral
     complexity    of    approximate  metric     capacitated    facility
location, where only the classic
variables are used (cf. Corollary \ref{cor:natural}). 
A similar idea was
independently used by Kaibel and Weltge in \cite{KaibelW14a} to derive
lower bounds on the number of  facets of a polyhedron which contains a
given integer set  $X$  and whose  set of integer  points is
$\operatorname{conv}(X)\cap \mathbb{Z}^d$.

We exhibit a concrete application of our 
methodology by proving in Theorem~\ref{theorem:mainCFL} an exponential lower bound on the
size  of  any   $O({N} )$-approximate mixed  product  relaxation  for  the
\cfl\ polytope, where $N$ is the
number of facilities in the  instance.
This result can be shown to
imply (cf. Theorem~\ref{cor:SA-CFL}) that the
$\Omega({N} )$-level SA relaxation for  \cfl, which is  obtained from any starting LP
of
  size $2^{o(N)}$ defined on the classic set of variables,    
has  unbounded gap $\Omega(N).$   Note, that it is well-known that
by  lifting only  the facility  variables, at  $N$ levels  the integer
polytope is obtained for \cfl\ \cite{BalasCC93}.   
This settles the open question
of \cite{AnBS13} whether there are LP relaxations upon which
the application of  lift-and-project methods captures  the strength of
preprocessing steps for \cfl. 
Our result establishes for the first time such a 
tradeoff for a {\sl universal} $\opn{SA}$
procedure that is independent of the starting relaxation $K.$
The proof follows the  methodology outlined above 
and is different from the standard arguments 
that apply  only  to  the  $\opn{SA}$ lifting  of  a specific  LP. 
Our earlier  $\opn{SA}$ construction  in \cite{KolliopoulosM14}
applied  the local-global method \cite{FernandezdlVKM07} that constructs  an appropriate  distribution of solutions  for each explicit constraint of the starting LP.

We leave as an open problem the extension of  the
equivalence between product and extended formulations 
from \zeone programs to mixed integer sets. 
We also believe that it would be of interest
to revisit  known extension-complexity lower bounds  using our
method, so as  to obtain simpler proofs.

\section{Preliminaries}  \label{sec:prelim}

$X \subseteq \mathbb{R}^d,$ is a {\em mixed integer set} if 
there is $p \in \{1, \ldots, d-1\}$
such that $d=n+p$  and $X \subseteq \{0, 1\}^n \times [0,1]^p.$ 
A {\em valid relaxation}  of the mixed integer set 
$X$ is any polyhedron $P$  such that $\opn{conv}(X) \subseteq P$.
 Given a valid relaxation $P$ of $X,$ such that $\opn{conv}(X) \cap ( \{0, 1\}^n \times [0,1]^p ) =  P \cap ( \{0, 1\}^n \times [0,1]^p)$, 
the \emph{level $k$ Sherali-Adams
($\opn{SA}$)  procedure}, $k\geq 1$,   is  as  follows \cite{SheraliA90}. 
Let $P$ be defined by 
the linear constraints $A{x}-{b}\leq 0.$ For every constraint $\pi(x)\leq 0$
of $P$, for every set of variables $U\subseteq \{x_i \mid i=1,\ldots,n\}$ such that
$|U|\leq k,$ and for every $W\subseteq U$, 
consider the {\em lifted} valid constraint:
$\pi(x)\prod_{x_i\in U-W}x_i \prod_{x_i\in W}(1-x_i)\leq 0$.
Linearize the system obtained this way by replacing (i) $x_i^2$ with
$x_i$ for all $i$ (ii) $\prod_{x_i\in I \subseteq [n]}x_i$
with  $x_I$  and (iii)  $x_k \prod_{x_i\in I \subseteq [n]}x_i,$ where
$k \in \{n+1, \ldots, d \}$ with $v_{Ik}.$ $\opn{SA}^k(P)$
is the projection of the resulting linear system onto the original
variables $\{x_1,\ldots,x_d\}.$ 
 We call $\opn{SA}^k(P)$ the relaxation {\em obtained from $P$ 
at level $k$} of the SA
hierarchy. 
It is well-known that $\opn{SA}^n(P) = \opn{conv}(X)$
(see, e.g., \cite{BalasCC93}). If $X$ is a \zeone set, i.e.,  $X \subseteq
\{0, 1\}^d,$ the above definitions hold mutatis mutandis and
$\opn{SA}^d(P) =  \opn{conv}(X).$

Given a polyhedron  $K(x,y) =  \{ (x,y) \in \mathbb{R}^d  \times
\mathbb{R}^{d_y} \mid Ax + By\leq b \}$ 
the {\em projection to the $x$-space}  is defined as  $\{x
\in \mathbb{R}^d \mid \exists y \in \mathbb{R}^{d_y} \colon  Ax + By\leq b
\}$ and is denoted as $\opname{proj_x}(K(x,y)).$
An {\em extended formulation (relaxation)} of a polyhedron $P(x) \subseteq \mathbb{R}^d$ is
a linear system  $K(x,y) =  \{ (x,y) \in \mathbb{R}^d  \times
\mathbb{R}^{d_y} \mid Ax + By \leq b \}$  such that
$\opname{proj_x}(K(x,y)) = P(x)$ ($\opname{proj_x}(K(x,y)) \supseteq P(x)$).
The {\em size} of a polyhedron is
the minimum number of inequalities in its halfspace description. The \emph{extension complexity}
of a polyhedron $P(x)$ is the minimum size of an extended formulation of $P(x)$.

We define now   $\rho$-approximate
formulations  as in  \cite{BraunFPS12}. 
Given a combinatorial optimization problem $T$, a {\em linear encoding} of
$T$ is  a pair $(L,O)$  where $L \subseteq  \{0, 1\}^*$ is the  set of
{\em feasible solutions} to the problem and $O\subset \mathbb{R}^*$ is the
set of {\em admissible objective functions.}  An instance of the linear
encoding is  a pair $(d,w)$ where  $d$ is a  positive integer defining
the dimension of the instance  and $w\subseteq O\cap \mathbb{R}^d $ is
the  set  of admissible  cost  functions  for  instances of  dimension
$d$. Solving  the instance  $(d,w)$ means finding  $x \in L  \cap \{0,
1\}^d$ such that  $w^{T}x$ is either maximum or  minimum, according to
the type of  problem $T.$  Let $P=\operatorname{conv}(
\{x \in \{0, 1\}^d \mid x \in L \})$ be the corresponding 
\zeone  polytope of dimension $d$.
Given a linear encoding $(L,O)$ of a maximization problem,  the
corresponding polytope $P,$  and $\rho \geq
1$,  a  $\rho$-\emph{approximate  extended formulation  of $P$ } is an
extended  relaxation  $Ax+By\leq b$ of $P$  with $x\in \mathbb{R}^d, y \in \mathbb{R}^{d_y}$
such that

\begin{align*}
\max \{w^{T}x \mid Ax+By \leq b \} \geq & \max \{w^{T}x \mid x \in P\} &
  \mbox{ for all }  w \in \mathbb{R}^d \mbox{ and }\\
\max \{w^{T}x \mid Ax+By\leq b \} \leq  & \rho \max \{w^{T}x \mid x \in
  P\} & \mbox{ for all }  w \in O \cap \mathbb{R}^d.
\end{align*}
For a minimization problem,  we require 
\begin{align*}
\min \{w^{T}x \mid Ax+By\leq b \} \leq & \min \{w^{T}x \mid x \in P\}   
& \mbox{ for all }  w \in \mathbb{R}^d \mbox{ and }\\
\min \{w^{T}x \mid Ax+By\leq b \}  \geq &  \rho^{-1} \min \{w^{T}x \mid x \in
P\} & \mbox{ for all }  w \in O \cap \mathbb{R}^d.
\end{align*}

The $\rho$-\emph{approximate  extension complexity} of \zeone integer
polytope $P(x) \subseteq [0,1]^d$ is the minimum size of a
$\rho$-approximate  extended formulation of $P.$ 
Given an 
extended formulation $Q(x,y)$ of $P(x),$  
a  {\em section}
of $Q$ is defined as  a vector-valued boolean function $g(x) \colon \{0,1\}^d
\rightarrow \mathbb{R}^{d+d_y}$ such that for
$x \in P(x) \cap \{0,1\}^d,$ $g(x)$ belongs to $Q(x,y)$ and $\opname{proj_x}(g(x))=x.$ Intuitively, the
section extends the encoding of solutions to the auxiliary
variables $y.$ Clearly, if  a particular  extended formulation
$Q$ has  been specified a priori, different  such
functions can be defined by filling in the last $d_y$ coordinates of 
$g(x_o)$ with a value from   $\{ y  \in \mathbb{R}^{d_y} \mid Ax_o + By \leq b \}.$
 
\begin{definition}
Given a  \zeone  integer polytope $P(x) \subseteq [0,1]^d,$ 
a {\em product relaxation
$D(z)$ of $P(x)$}  is an extended relaxation $D(z)$ of $P(x),$ 
where $z \in \mathbb{R}^{2^{d}-1}$ and for every nonempty subset
${\mathcal{E}} \subseteq \{x_1, x_2, \ldots, x_{d} \}$ of the original variables,
we have a variable $z_{\mathcal{E}},$  (where  $z_{\{x_i\}}$ denotes
$x_i,$ $i=1,\ldots,d$), and there is a section $f(x)$ of $D$ s.t.
the  corresponding coordinate of $f$ at $\mathcal{E}$
is 
$f_{\mathcal{E}}(x) = \prod_{x_i \in \mathcal{E}}  x_i.$ We refer to this
function $f$ as {\em the product section}. 
\end{definition}

Let  $f$ denote the product section. Define the {\em
  canonical product relaxation of $P$} as 
$\hat{D} = \opname{conv} \{ f(x) \mid x \in P(x) \cap \{0,1\}^{d_x}
\}.$  The polytope $\hat{D}$ corresponds to the ``tightest'' possible
product relaxation. 

For a mixed integer set $M(x,w) \subseteq \{0,1\}^{d_x}
\times \mathbb{R}^{d_w}$ the corresponding mixed integer polytope
$P(x,w)$ is $\opname{conv}(M(x,w)).$ 
In case one starts from a mixed integer polytope, the additional $z$
variables of the product relaxation correspond to sets that
contain at most one fractional variable. Including only one fractional
variable in each product,  mimics the variable space of the
final-level SA relaxation. 

\begin{definition}
Let   $P(x,w) \subseteq [0,1]^{d_x}
\times \mathbb{R}^{d_w}$ be a mixed integer  polytope. A  {\em mixed product relaxation
$D(z)$ of $P(x,w)$}  is an extended relaxation $D(z)$ of $P(x,w),$ 
where $z \in \mathbb{R}^{(d_w+1)2^{d_x}-1 }$, with  $z_{\{w_j\}}=w_j$, $j=1,\ldots,d_w$, 
and \\ (i)  for every  set
$\varnothing \neq {\mathcal{E}} \subseteq \{x_1, x_2, \ldots, x_{d_x}
\}$
we define
$d_{w}+1$ variables: one that we denote 
$z_{\mathcal{E}}$ and, for each fractional variable $w_j,$
$j=1,\ldots,d_w,$ one that we denote $z_{\mathcal{E}w_j}.$
Moreover $z_{\{x_i\}}$ denotes $x_i$,  $i=1,\ldots,d_x$.   \\
(ii) there is a section $f(x,w)$ of $D$  s.t.
the corresponding coordinates of $f$ are
 $f_{\mathcal{E}} (x,w) = (\prod_{x_i  \in \mathcal{E}}  x_i)$ and, 
for each variable $w_j,$
$j=1,\ldots,d_w,$  $f_{\mathcal{E}w_j} (x,w)
 = (\prod_{x_i  \in \mathcal{E}}  x_i)\cdot w_j.$
We refer to this
function $f$ as {\em the mixed product section}. 
\end{definition}

\noi
 The {\em
  canonical product relaxation of $P(x,w)$} is similarly 
defined  as 
$\hat{D} = \opname{conv} \{ f(x,w) \mid (x,w)  \in P(x,w) \cap
\left (\{0,1\}^{d_x} \times \mathbb{R}^{d_w} \right ) \}.$

Note that the lifted polytope produced by the $d$-level ($d_x$-level) 
Sherali-Adams procedure
applied on some specific linear relaxation of the $0$-$1$ polytope $P(x)$
(mixed integer $P(x,w)$), after linearization and before projection to
the original variables, 
is a (mixed) product relaxation.


\section{The  expressive power of product relaxations}

In this section we show the following. For every \zeone polytope 
$P(x)$
and every (approximate) extended formulation 
$Q(x,y) = \{(x,y) \in \mathbb{R}^{d_x} \times \mathbb{R}^{d_{y}} \mid Ax + By \leq b \} $ 
of  $P(x)$    
there  is   a  product
relaxation $T[Q(x,y)]$ whose size is at most that of $Q(x,y)$ and 
is at least as strong.

   A 
\emph{substitution} $T$ is  a linear map of the form $y=T
z$  where  $T$  is  a  $d_y \times  (2^{d_x}-1)$  matrix  and  $
z$  is a  $2^{d_x}-1$ dimensional  vector having  a coordinate
$z_{\mathcal{E}}$ for each nonempty set $\mathcal{E}$ of the form $\{ x_i 
\mid i \in S \subseteq 2^{\{ 1,...,d_x \}} \}.$ For any substitution 
$T,$  
the {\em translation} 
of $Q(x,y),$ denoted ${T}[Q(x,y)],$ the formulation resulting
by 
substituting $T_{(i)} z,$ for $y_i,$ $i= 1,...,d_y.$ Here  $T_{(i)}$
denotes  the $i$th  row of $T.$  
If in addition ${T}[Q(x,y)]$ is a
product relaxation of $P(x)$ we say that it is a {\em translation of $Q$ to
  product relaxations} (recall that the original variables $x_i$ coincide with the variables $z_{\{x_i\}}$). 
Observe that the number of  inequalities of ${T}[Q(x,y)]$ is the same as
in $Q(x,y).$ 
The translation may heighten exponentially the dimension,
but  since our  methodology  will give lower  bounds  on the  size of  the
product relaxations 
those bounds apply
to the size of $Q(x,y)$ as well.  

\begin{theorem} \label{theorem:ext1}
Given a \zeone polytope 
$P(x) \subseteq [0,1]^{d_x},$ for every polytope $Q(x,y)$ such that
$P(x)\subseteq \opname{proj_x}(Q(x,y))$  
there is a translation ${T}[Q(x,y)]$ to product
relaxations such that
$P(x) \subseteq \opname{proj_x}( {T}[Q(x,y)]) \subseteq \opname{proj_x}(Q(x,y)).$ 
\end{theorem}

\begin{proof}

We shall give a substitution 
 $T$  for the variables $y \in \mathbb{R}^{d_y}$ 
of $Q(x,y)$ so that the theorem holds. 
Let $g(x)$ be a section of $Q(x,y)$ (recall that a section associates  every feasible \zeone
vector $x$ of $P(x)$ to a specific  $y$  such that $(x,y) \in Q(x,y)$).

Observe that the coordinates of the product section plus the constant $1$
 correspond exactly to the monomials of the Fourier basis. We
denote by $(p,1)\in \mathbb{R}^{n+1}$ the vector resulting from $p \in \mathbb{R}^n$ by
appending the scalar $1$ as an extra coordinate. 
By basic functional analysis (see, e.g.,
\cite{DBLP:books/daglib/0028687}), 
there is a $d_y\times 2^{d_x}$ matrix $A$ such that

\begin{equation}
g(x)=A\cdot(f(x),1) \label{eq:fourier}
\end{equation}

We  define the substitution $T$  by linearizing  the above equation; 
we replace  the sections $g$ and $f$ with the corresponding variable
vectors $y$ and $z$ (recall $z$ is the product vector) to obtain:

\begin{equation} \notag
y=A\cdot(z,1).
\end{equation}
 
Obviously
$\opname{proj_x} ( T[Q(x,y)] )  \subseteq \opname{proj_x} ( Q(x,y) )$:
from any  feasible solution  $(x_0, z_0)$ of $T[Q(x,y)]$ we  can derive  a feasible
solution $(x_0,y_0)$ of  $Q(x,y)$  by  setting $y_0$ equal to $Az_0.$

We will now show that  $ P(x) \subseteq \opname{proj_x} ( T[Q(x,y)]).$
It suffices to show that for every $x' \in P(x)\cap
\{0,1\}^{d_x}$ the vector $f(x')$ is feasible for $T[Q(x,y)]$ as 
required  by the  definition of  product relaxations.  Observe  that by
letting the $z$ vector take the  values $f(x'),$ by \eqref{eq:fourier} we
get that  the quantities involved  in the inequalities  of $T[Q(x,y)]$
are  the   exact  same   quantities  involved  in   the  corresponding
inequalities of  $Q(x,y)$ for $g(x')$.  But by the  definition of
section,  $g(x')$ is feasible  for $Q(x,y)$  and thus  $f(x')$ is
feasible for $T[Q(x,y)]$.

\end{proof}

\begin{corollary}  \label{cor:ext1}
A lower bound $b$ on the size of any product
relaxation  $D$ which is a $\rho$-approximate extended formulation of 
 the \zeone polytope $P(x)$, for $\rho \geq 1$, implies a lower bound $b$ on the size 
 of any $\rho$-approximate  extended  formulation $Q(x,y)$ of $P(x)$.
\end{corollary}

\opt{short}{
}


Let $P(x,w)$ be a mixed integer polytope. The notion of the section 
  of $P$  for some  extended relaxation
$Q(x,w,y)$ of  $P$ is more challenging. Intuitively,  the solutions are
characterized by two parts --  a boolean part of the $0-1$ assignments
on the integer variables $x$ and a "linear" part of the real variables
$w$ in the following sense: once  the boolean part (the "hard" one) is
fixed, the linear part can be obtained as the feasible region of a
(usually small) system of inequalities, possibly empty.

Motivated  by the  above we  define the  following type  of sections
for  an  extended  formulation  $Q(x,w,y)$ of a mixed-integer
polytope. A {\em mixed-linear section } of EF $Q$ is a section $g$ for
which  
at  variable
$y_i$ the value $g_i(x',w)$ for a given 
integer vector $x'$ is an affine  function on $w$ denoted $g^{x'}_{i}(w).$ 
If there is such a mixed-linear section for  $Q(x,w,z),$ 
we say that $Q$ is an {\em extended formulation with a  mixed
 linear section}.
An example of EFs with a mixed linear section
are formulations arising from the  SA procedure where $y$ is the vector
of the new variables corresponding to the linearized products. 
The following theorem can be proved similarly to
Theorem \ref{theorem:ext1}.

\begin{theorem}\label{theorem:ext2}
Given a mixed integer polytope $P(x,w) \subseteq  [0,1]^{d_x}
\times \mathbb{R}^{d_w},$ 
for every $\rho$-appro\-xi\-mate, $\rho \geq
1,$   extended formulation $Q(x,w,y)$  with a 
mixed linear section, there is a translation $T[Q(x,w,y)]$ 
 to mixed product
relaxations such that
$$P(x,w) \subseteq \opname{proj_{x,w}}(T[Q(x,w, y)]) \subseteq \opname{proj_{x,w}}(Q(x,w,y)).$$ 
\end{theorem}

\begin{proof}
Let  the dimension  of $P(x,w)$  be  $d=d_x+d_w$.  We  shall give  for
the variables $y \in \mathbb{R}^{d_y}$ of $Q(x,w, y)$ a substitution $T$ 
so that the theorem holds. 

Consider a variable $y_i$ and the corresponding coordinate 
of the mixed linear 
section, 
$g^{x'}_{i}(w)=\sum_j b^{x'}_{i}  w_j + c_{x'}$ for each $x' \in
\{0,1\}^{d_x}$ and $i=1,\ldots,d_y.$ 


First, we will prove a helpful claim which states a fact from 
elementary Fourier
analysis in our setting. 
For $x, s \in \opname{proj_x} \left ( P(x,w) \cap (\{0,1\}^{d_x}\times
\mathbb{R}^{d_w}) \right ),$ 
define the boolean
indicator operator $\chi_{s}(x)$ to be $1$ when $s=x$ and $0$ otherwise. 
First, we will show that this operator can be expressed as a linear
combination of the product sections constrained to monomials with
only boolean variables.  
In other words, we determine 
coefficients $a^s_{\mathcal{E}},$ $\mathcal{E} \subseteq \{x_1,
  \ldots, x_{d_x}\},$   such that 
$\chi_{s}(x) = \sum_{\mathcal{E}}
a^s_{\mathcal{E}} f_{\mathcal{E}}(x).$
The  translation of the  indicator operator  $i_{s}(x)$ of  an integer
solution  $s$  is a  linear  expression  of  the form  $T_{i_s}=  \sum
a^s_{\mathcal{E}}P[\mathcal{E}](x)$. 
We shall iteratively generate the
coefficients $a^s_{\mathcal{E}}$.  
The only nonzero coefficients 
will be those corresponding to sets of variables that are supersets
of the set of variables being $1$ in $s$ -- let that set be
$\mathcal{E}^s_1$. We give  the construction iteratively starting from
$|\mathcal{E}^s_1|$ to $d_x,$ defining in step $k$ the coefficients of
such sets of size $k$.

In the  first iteration  simply set $a_{\mathcal{E}^s_1}=1$.  At step
$k> |\mathcal{E}^s_1|$, for each set $\mathcal{E}'$ of size $k$ that
is  a  superset  of  $\mathcal{E}^s_1$,  set  $a^s_{\mathcal{E}'}=  -
\sum_{\mathcal{E}  \subset  \mathcal{E}'}  a^s_{\mathcal{E}}$.  This
concludes the definition of the coefficients.

\begin{claim} \label{claim:operator}
For each integer solution $s'  \in \opname{proj_x} \left ( P(x,w) \cap (\{0,1\}^{d_x}\times
\mathbb{R}^{d_w}) \right )$, $\chi_s(s') = \sum_{\mathcal{E}} 
a^s_{\mathcal{E}} f_{\mathcal{E}}(s')$.
\end{claim}

\noindent
{\em Proof of the claim.} 
By overloading the notation, we denote by $s$ both the integer
solution and the support of that integer solution, that is the set
$\{x_i \mid s_i=1\}$. If $s' \supseteq s$ then the  nonzero terms of
the sum $\sum_{\mathcal{E}} 
a^s_{\mathcal{E}} f_{\mathcal{E}}(s')$ are exactly those 
that correspond to sets $\mathcal{E}$ such that 
$s \subseteq \mathcal{E} \subseteq s'$. We have that $\sum_{\mathcal{E}} 
a^s_{\mathcal{E}} f_{\mathcal{E}}(s')= \sum_{\mathcal{E} \subseteq s'}
a^s_{\mathcal{E}}$ which, by the construction of the  coefficients, 
is $1$ if $s=s'$ and $0$ if $s'\supset s$, as required. 
Otherwise, if $s-s' \neq \varnothing$, then all the $f_{\mathcal{E}}(s')$ 
with nonzero coefficients are $0,$ so $\sum_{\mathcal{E}} 
a^s_{\mathcal{E}} f_{\mathcal{E}}(s')=0$.

By Claim \ref{claim:operator} we have that for an integer vector $s \in \{0,1\}^{d_x}$ 
 the  indicator operator $\chi_{s}(x)$ is equal to 
 $\sum_{\mathcal{E} \subseteq \{x_1, \ldots, x_{d_x}\}}
a^s_{\mathcal{E}} f_{\mathcal{E}}(x).$
For each set of integer variables $\mathcal{E}$ and each fractional
variable $w_j$ let $z_{\mathcal{E}w_j}$ denote the corresponding 
mixed product variable and $f_{\mathcal{E}w_j}(x,w)$ the 
corresponding coordinate of the mixed product section.  
It is now easy to show the following.  

\begin{claim} \label{claim:mixedtrans}
For each mixed integer solution $(x',w')$, and for $i=1,\ldots,d_y,$\\
$g^{x'}_{i} (w')= \sum_j \sum_{\mathcal{E}}a^s_{\mathcal{E}}  b^{x'}_{i}
f_{\mathcal{E}w_j}(x',w') + \sum_{\mathcal{E}}a^s_{\mathcal{E}} c_{x'}f_{\mathcal{E}}(x').$
\end{claim}

To conclude the definition of $T$, set 
\[ y^i =\sum_{x'}\sum_j \sum_{\mathcal{E}}a^s_{\mathcal{E}}  b^{i}_{x'}
z_{\mathcal{E}w_j}   + \sum_{x'} \sum_{\mathcal{E}}a^s_{\mathcal{E}}  c_{x'}
z_{\mathcal{E}} , \;\; i=1,\ldots,d_y.\]
which implies 
\[ y^i =\sum_{x'}\sum_{\mathcal{E}}a^s_{\mathcal{E}}  (\sum_j
b^{i}_{x'}z_{\mathcal{E}w_j} + c_{x'}z_{\mathcal{E}}), \;\;
i=1,\ldots,d_y   \]
By  Claim~\ref{claim:mixedtrans}, using arguments  similar to the ones
in the proof of Theorem~\ref{theorem:ext1}, it follows that 
that $P(x,w)\subseteq \opname{proj_{x,w}} (T[Q(x,w,y)] ) \subseteq 
\opname{proj_{x,w}} ( Q(x,w,y)).$
\end{proof}

\begin{corollary}  \label{cor:ext2} 
A lower bound $b$ on the size of any mixed product
relaxation  $D$  which is a $\rho$-approximate extended formulation of 
 the \zeone mixed integer polytope $P(x,w)$ implies a lower bound $b$ on the size 
 of any $\rho$-approximate  extended  formulation $Q(x,w,y)$ of
 $P(x,w)$ with a  mixed linear section. 
\end{corollary}


\section{A method for lower bounding the size of LPs with known sections}

Here we present  a methodology 
to lower bound  the size  of  relaxations that
achieve a desired integrality gap.  For simplicity we do not deal in this
section with mixed integer sets.

Our method can be summarized as follows.  
Let  $G(z) \subseteq [0,1]^d$  be a \zeone polytope. 
We design a family $\mathcal{I}$ of instances parameterized by the
dimension $d.$ For  each instance $I\in  \mathcal{I}$ of dimension  $d$ we
define a set of  points $\mathcal{C}_I \subseteq [0,1]^{{d}} \setminus G(z)$ which we call
the   \emph{core    of  $I$ with respect to $G.$}
    Note that the points of the  core must
be infeasible for $G.$   To prove a lower
bound $r(n)$ on the size of $G$ it suffices to  show  that at least that many
inequalities are needed to separate $\mathcal{C}_I$ from $G.$
Additionally, for a minimization problem with  $O$ being the set of admissible
objective functions, if for some
$z \in \mathcal{C}_I$ there is
an admissible cost function $w_z $ such that $w_z^T z < \rho^{-1} 
\text{Opt}_{I,w_z},$ $0< \rho \leq 1,$ 
where $\text{Opt}_{I,w_z}$ is  the cost of the optimal  integer solution with
respect to $w_z,$ we call $z$ {\em $\rho$-gap inducing wrt $O.$} If we design the
core so that all its members are $\rho$-gap inducing, the lower bound
will  hold for $\rho$-approximate formulations.

To define constructively the core for a specific family of extended
formulations of a polytope $P$ the sections of the variables $z$ must be known.  
This  requirement is  fulfilled by  the product  relaxations we
will focus on. 
By Theorem 
\ref{theorem:ext1} 
above, proving a lower bound
on the size for an arbitrary extended relaxation $Q(x,y)$ of a
polytope $P(x)$ can be reduced  to a proof of the same bound on
the size of a corresponding product relaxation $D(z)$. The
following meta-theorem shows that such a proof can always be obtained
by proving  the existence  of a suitable  core for  the product
relaxation. 
Recall the  definition of the ``tightest'' product relaxation of $P(x),$ $\hat{D},$
in Section~\ref{sec:prelim}. 
We say that a set of vectors $s \subseteq [0,1]^{d} \setminus \hat{D}$ is \emph{conflicting} if
$\operatorname{conv}(s) \cap \hat{D} \neq \varnothing.$
Any single valid inequality of  $\hat{D}$ cannot separate all
points of a conflicting set. 
Given a set  $O_d \subseteq \mathbb{R}^d$ of admissible objective functions
associated with a \zeone polytope $P(x) \subseteq [0,1]^d,$ we define 
$\tilde O_d \subseteq \mathbb{R}^{2^d-1},$ to contain the vectors in $O_d$
extended with zeroes in the coordinates corresponding to the non-singleton
product variables. 

\begin{theorem} \label{theorem:meta} 
Given a \zeone polytope $P(x) \subseteq [0,1]^d,$ and an associated set of
admissible objective functions $O_d \subseteq \mathbb{R}^d,$ 
the $\rho$-approximate extension complexity, $\rho \geq 1,$ of $P(x)$
 is
at least $r(n),$ iff there exists a family of instances 
$\mathcal{I}(n)$ and, for every $I \in
\mathcal{I},$  a core $\mathcal{C}_I$ wrt $\hat{D},$ which consists 
of $\rho$-gap inducing
  vectors wrt  $\tilde O_d,$   with the  following property:
for  any  partition  of $\mathcal{C}_I$  into less  than  $r(n)$ parts
there must be a part containing a set of conflicting vectors.
\end{theorem}
\begin{proof}
Assume first that the $\rho$-approximate extension
complexity is at least $r(n).$ 
Define $\mathcal{C}_I$ to be the set 
of all $\rho$-gap  inducing
product vectors.  
If  we  can   partition  $\mathcal{C}_I$ into less  than  $r(n)$ parts
so  that there is
no conflicting subset $s$ in any part, then we can define an
inequality for each  part of the partition that  separates the vectors
of at least that part from $\hat{D}.$ But we know that 
less than $r(n)$ inequalities  cannot separate all  the $\rho$-gap
inducing product  vectors. Thus we have that  for any decomposition
of  those vectors  into less than $r(n)$   parts there  must be  a part
containing a set of conflicting vectors. 

Conversely, assume we   can find
a core $\mathcal{C}_I$ wrt $\hat{D}$ consisting of $\rho$-gap inducing
vectors such that 
for any partition of $\mathcal{C}_I$ into less than $r(n)$ sets
there must be a part containing a set of conflicting vectors. Then  the size
of $\hat{D}$ is at least $r(n).$ 
If not, there is a decomposition into less than $r(n)$ parts where each
part consists of the core members separated by each inequality -- in case a  member is separated by more than one inequality, we arbitrarily include it into just one of the resulting parts. 
Observe that $\mathcal{C}_I$   is not
only a core wrt $\hat{D}$ but also is 
a  core wrt any $\rho$-approximate product relaxation of
$P.$ By Theorem \ref{theorem:ext1}, the lower bound $r(n)$ applies to
the size of any $\rho$-approximate extended formulation of $P.$ 
\end{proof}

Let $\mathcal{H}(\mathcal{C}_I)$ be the, possibly infinite, 
 hypergraph with vertices the members of $\mathcal{C}_I$ and hyperedges the
 conflicting subsets of $\mathcal{C}_I$. Theorem \ref{theorem:meta} can be
 restated more conveniently:

\begin{theorem} \label{theorem:hyper}
Given a \zeone polytope $P(x) \subseteq [0,1]^d,$ and an associated set of
admissible objective functions $O_d \subseteq \mathbb{R}^d,$ 
the $\rho$-approximate extension complexity, $\rho \geq 1,$ of $P(x)$
 is
at least $r(n),$ iff there exists a family of instances 
$\mathcal{I}(n)$ and, for every $I \in
\mathcal{I},$  a core $\mathcal{C}_I$ wrt $\hat{D},$ which consists 
of $\rho$-gap inducing
  vectors wrt  $\tilde O_d,$  such that $\mathcal{H}(\mathcal{C}_I)$ has
  chromatic number $r(n)$.
\end{theorem}

Theorem~\ref{theorem:meta} suggests that the best possible lower bound
on the extension complexity can always be achieved by proving the existence of an 
 appropriate
core 
in the product space.
In the applications in this paper we implement a  
version of the method 
that imposes  stronger requirements
 on the decomposition, namely  the constructed hypergraph will be a clique.

\section{Lower bounds for approximate mixed product relaxations
  for \cfl}\label{section::distr-cfl}
For \cfl, the
linear encoding $\mathcal{N}_{\text{\cfl}} = (L, O)$ is defined as follows. For a \cfl\ instance, given
the number $n$ of facilities, the number $m$ of clients, the
capacities    $K\in   \mathbb{R}^n_+$    and    the   demands    $D\in
\mathbb{R}^m_+$, 
we use the classic variables $y_i,$
$i=1,\ldots,n,$  $x_{ij},$  $i=1,\ldots,n,$  $j=1,\ldots,m$  with  the
usual  meaning of facility opening and client assignment respectively. 
The set of feasible solutions $(y,x)$ is defined 
in the obvious manner. 
Thus for dimension $d=n + nm,$ 
$L \cap \{0,1\}^d$ is completely determined by the quadruple
$(n,m,K,D).$ The
set of admissible objective functions $O \cap \mathbb{R}^{n+nm}$ is the set of pairs
$({\bf f}, {\bf c})$ where ${\bf f} \in \mathbb{R}^n_+$ are the facility
opening costs and ${\bf c}=[c_{ij}] \in \mathbb{R}^{nm}_+$ are 
 connection costs  that satisfy  $c_{ij}\leq
c_{i'j}+c_{i'j'}+c_{ij'}$.   

The capacitated facility location  problem with general capacities and
demands is  a mixed integer optimization problem  where the facilities
are  opened integrally  but the  clients  are allowed  to be  assigned
fractionally  to   the  set  of   opened  facilities.
 In this section, 
we show an exponential lower bound on the size of any
mixed product relaxation  of the \cfl\  polytope. 

In our proof we will consider a parameterized  instance $I=I(3n,m,U,d)$  
with uniform capacities $U$ and uniform unit
demands $d=1,$ where  
$3n$ is the number of facilities, and $m$ the number of clients. 
 Furthermore we will have 
that the number of clients is $m=n^4+1$ and the capacities and demands are such that
  $(n^4+1)-nU=2^{-n^2}$. Observe that $n^3<U<(n^3+1)$.
In order to define the  core $\mathcal{C}_I$ of the instance $I$ we
first describe a random experiment based on whose outcome we will
later define the members of the core. 
Given disjoint sets $k,l \subseteq  F$ of size $n$ each, the random experiment defines 
 a distribution $\mathcal{D}_{k,l}$ over mixed integer vectors in the
 classic encoding. These vectors correspond in general to pseudo-solutions. 
The following
experiment  defines the  distribution $\mathcal{D}_{k,l}.$
The quantities $\bar{x}_{ij}$ are defined in Lemma~\ref{lemma::exp-vector} below. 
\\ 
\noi

\centerline{\hrulefill {\sc Random Experiment} \hrulefill}
\noi 
Facilities in $k$ are always opened. 

\noi  {\it Case  1}.  With  probability  $1-\frac{20}{n^2(1+1/n)}$ all
facilities in $F-l$ are opened and those of $l$ are closed. Distribute
evenly the client  demand to facilities in $k$.  Note that this outcome
of  the  experiment    does  not  respect  the
capacities.

\noi  {\it  Case 2.} 
Otherwise, with probability  $\frac{20}{n^2(1+1/n)}$ pick at random a
 subset $q$ of the facilities in $F-k$ with at least one facility from
 $l$ and  open them.  Assign randomly demand  to each facility  $i$ in
 $q\cap l$  so that  $i$ takes $\frac{\sum_j  \bar{x}_{ij}}{10/n^2}$ units
 and the rest  of the demand is equally  distributed to the facilities
 in $k$.
\noi
\hrulefill

\begin{lemma}\label{lemma::exp-vector}
The expected vector $(\bar{y},\bar{x})$ wrt $\mathcal{D}_{k,l}$ is the following: 
$\bar{y_i}=1$ for $i\in k$, $\bar{y_i}=1-\frac{10}{n^2(1+1/n)}$ for
 $i \in  F-k-l$,  $\bar{y_i}=\frac{20(2^{n-1})}{n^2(1+1/n)(2^n-1)}$
 for $i \in  l$. For all $j \in C,$ $\bar{x}_{ij}=\frac{1-n^{-2}}{|k|}$ for $i \in k$, $\bar{x}_{ij}=0$ for $i \in F-\{k\cup l\}$, $\bar{x}_{ij}=\frac{n^{-2}}{|l|}$ for $i \in l$.
\end{lemma}

\opt{full}{
\begin{proof}
For $i \in k$ we have that $i$ is always open in $\mathcal{D}_{k,l}$
so $\bar{y}_i=P_{\mathcal{D}_{k,l}}[i \mbox{ \small{opened}}]=1$.
 For $i\in l$, note that it is opened only in case $2$ when $i \in q$. The set $l \cap q$ is a randomly selected nonempty subset of $l$. Thus $\bar{y}_{i}=P_{\mathcal{D}_{k,l}}[i \mbox{ \small{opened}}]=P_{\mathcal{D}_{k,l}}[i \in q \mbox{ \small{of case 2}}]=\frac{20}{n^2(1+1/n)}\frac{2^{n-1}}{2^n-1}=\frac{20(2^{n-1})}{n^2(1+1/n)(2^n-1)}$. Similarly for the $\bar{y}$ variables of facilities in $F-(k\cup l)$.  As for the assignment variables, for each $j$ and each facility $i \in l$, each time $i$ is opened it is assigned $\frac{\sum_j \bar{x}_{ij}}{\bar{y}_{i}}$ demand at random and since it is opened a $\bar{y}_i$ fraction of the time, the total expected demand assigned to it is $\frac{\sum_j \bar{x}_{ij}}{\bar{y}_{i}}\bar{y}_{i}=\sum_j \bar{x}_{ij}$. Since the assignments are random each client is assigned to $i$ with the same fraction in expectation, so $P_{\mathcal{D}_{k,l}}[i \mbox{ \small{assigned to} } j]=\bar{x}_{ij}$.
 Facilities in $F-\{k\cup l\}$ are never assigned any demand. By the construction of the distribution the demand not assigned to $l$, is assigned to the facilities in $k$ randomly and so the expected $\bar{x}_{ij}$ have their intended values.
\end{proof}
} 

The  distribution  $\mathcal{D}_{k,l}$ will  be  subsequently used  to
define the members of the  core $\mathcal{C}_I$.   Let 
$\mathcal{E}$ be a subset of  integer variables in the original
space, i.e., $\mathcal{E} \subseteq \{y_1, \ldots, y_{3n}\}.$ 
We denote by $E_{\mathcal{D}_{k,l}}[\mathcal{E}]$ the expectation of
the event  where all  the variables in $\mathcal{E}$  have value $1,$
i.e., the expectation of  the product $\prod_{y_{i_k} \in \mathcal{E}}
y_{i_k}.$ Similarly, we denote
$E_{\mathcal{D}_{k,l}}[\mathcal{E}x_{ij}]$
the {expectation} of
the product $(\prod_{y_{i_k} \in \mathcal{E}} y_{i_k}) \cdot x_{ij}.$  
Let $\chi(case1), \chi(case2)$  be the \zeone random
variables that indicate whether Case 1 and Case 2 
occur, respectively. We denote by $E_{\mathcal{D}_{k,l}}[\mathcal{E}
  \cap case1]$ the expectation of  the product $(\prod_{y_{i_k} \in
  \mathcal{E}}y_{i_k}) \cdot \chi(case1)$ and by
$E_{\mathcal{D}_{k,l}}[\mathcal{E}x_{ij} \cap case1]$ the {expectation} of
the product $(\prod_{y_{i_k} \in \mathcal{E}} y_{i_k}) \cdot x_{ij}
\cdot \chi(case1).$  Similarly for Case 2.
 Intuitively,  $E_{\mathcal{D}_{k,l}}[\mathcal{E}x_{ij} \cap case1]$ is 
the "mass" that $\mathcal{D}_{k,l}$ assigns to $x_{ij}$ over all outcomes
 of case 1 where the variables of $\mathcal{E}$ have value $1$.

To   simplify   notation,   we   use
$z(i)$ instead of $z_{i}$ to refer to a coordinate
of vector $z$ indexed by $i.$ From now on, $P$ denotes the
\cfl\ polytope and $\hat{D}$  its
canonical product relaxation.

\begin{definition}
Fix a set $k \subset F$ of size $n$.
The \emph{core} $\mathcal{C}_I$ 
 of the  instance $I(3n,n^4+1,U,1)$ wrt $\hat{D}$ is  the following
 set of product vectors:  $\forall l  \subset F$  with $|l|  =  n$ and
 $k\cap l =\varnothing$ and for every set $\mathcal{E}$ of integer
 variables   and for every fractional variable $x_{ij}$ 
we define $z_{k,l}(\mathcal{E})=E_{\mathcal{D}_{k,l}}[\mathcal{E}]$
and $z_{k,l}(\mathcal{E}x_{ij}) =
E_{\mathcal{D}_{k,l}}[\mathcal{E}x_{ij}].$  
\end{definition}

Now we  are ready to state the  key Lemma~\ref{mainlemma} from which  
our main theorem
will   be  derived.  
\opt{short}{The proof of the lemma is quite
  technical. A sketch is presented in Subsection~\ref{section:keylemma-sketch}.
}
\opt{app}{
The proof of the lemma is quite
  technical and is deferred to the Appendix. A sketch is presented in Subsection~\ref{section:keylemma-sketch}.
}

\opt{full}{
The    proof   of    the    lemma    is   in    Section
\ref{section:keylemma}.
}

\begin{lemma}\label{mainlemma}
For  any two $z_{k,l},z_{k,l'}\in  \mathcal{C}_I$ such  that $l-l'\neq
\varnothing$         there        is        some         $z        \in
\operatorname{conv}(z_{k,l},z_{k,l'})$    which   is    feasible   for
$\hat{D}$.
\end{lemma}

\begin{theorem}\label{theorem:mainCFL}
Given  the  family  of \cfl\ 
instances  $I(3n,n^4+1,U,1),$  each member of $\mathcal{C}_I$ is 
$\Omega(n)$-gap inducing and
$\chi(\mathcal{H}(\mathcal{C}_I))=2^{\Omega \left ( n 
  \right )}$. Therefore, there is a constant $c  >0,$ s.t. any $c  N$-approximate 
  EF for \cfl\ with a mixed linear section
has size $2^{\Omega(N)},$ where $N$ is the number of facilities. 
\end{theorem}

\begin{proof}
Since we proved  in Lemma \ref{mainlemma} that any  two members of the
core $\mathcal{C}_{I}$  form a conflicting set, $\mathcal{H}(\mathcal{C}_I)$ is a clique and thus its chromatic number is $|\mathcal{C}_{I}|=$${2n}\choose
{n}$$=2^{\Theta(n)}$.  For  each member  of  the  core $z_{k,l}$ there is  an
admissible cost function $w_{k,l}$ inducing $\Theta(n)$ gap: facilities in $l$ have unit opening costs
and every other facility has $0$ opening cost. The facilities in
$k\cup l$ and all the clients are co-located, and the rest of the
facilities are co-located at distance $2^{n^2}$ from the former. Observe
that each feasible mixed integer solution has a cost of at least $1$
since either some facility in $l$ must be opened integrally or at
least $2^{-n^2}$ client demand has to be assigned to some facility in $F-k-l$. On the
other hand the cost of $z_{k,l}$ wrt $w_{k,l}$ is $\Theta(n^{-1})$
since the $(y,x)$ projection of $z_{k,l}$ is the expected vector
$(\bar{y},\bar{x})$ of $\mathcal{D}_{k,l}$. 
\end{proof}

On the other hand, it is easy to see that for every instance $I$ of \cfl\, there is an exact formulation of size $2^{N}p$ where $p$ is a polynomial expression in the size of the instance.
Moreover this formulation can be written as a mixed product relaxation. 
The idea is to simply define a formulation for each choice of the opened facilities and then take the convex hull of those polytopes.
\opt{app}{Please refer to the Appendix for more details.}

\begin{observation}  \label{obs:exact} 
There is an exact mixed product relaxation of the \cfl\ polytope of size
$2^{N}p,$ where 
$p = \Theta (mN),$ $N$ and $m$ being the number of facilities and clients respectively.
\end{observation}

\opt{full}{
\begin{proof}
For each choice of opened facilities $O \subseteq F$ consider the following polytope $P^O$:

\begin{align}
x^O_{ij} \leq y^O_i   &&  \forall i \in F, \forall j \in C \\
\sum_{i \in F} x^O_{ij} =1   &&   \forall j \in C \\
0 \leq x^O_{ij} \leq 1  &&  \forall i \in F, \forall j \in C \\
\sum_{j \in C} x^O_{ij} \leq U_i y^O_i   &&   \forall i\in F \\
y^O_i =0   && \forall i \in F-O \\
y^O_i =1   && \forall i \in O   
\end{align}

It obviously an exact formulation of the (possibly empty) polytope when we fix the values of the $y_i$'s wrt $O$. Then by introducing "selection" variables $s^O$ with the meaning of whether the set of opened facilities is $O$ or not we get the following convex combination of polytopes:

\begin{align}
\sum_{O \subseteq F}s^O=1\\
s^O\geq 0 && \forall O \subseteq F\\
x^O_{ij} \leq y^O_i   &&  \forall i \in F, \forall j \in C,\forall O \subseteq F \\
\sum_{i \in F} x^O_{ij} =s^O   &&   \forall j \in C,\forall O \subseteq F \\
0 \leq x^O_{ij} \leq s^O &&  \forall i \in F, \forall j \in C,\forall O \subseteq F \\
\sum_{j \in C} x^O_{ij} \leq U_i y^O_i   &&   \forall i\in F, \forall O \subseteq F \\
y^O_i =0   &&\forall O \subseteq F, \forall i \in F-O \\
y^O_i =s^O   &&\forall O \subseteq F, \forall i \in O   
\end{align}

The proof is concluded by observing that the above linear program has a  mixed linear section and 
by using Theorem \ref{theorem:ext2}. 
\end{proof}
}

\begin{theorem}\label{cor:SA-CFL}
Let $P$  be 
any linear  relaxation 
of the  \cfl\ polytope for the family of instances $I(3n,n^4+1,U,1)$ that uses the  encoding
$\mathcal{N}_{\text{\cfl}}$ and has size $2^{o(n)}.$ 
There is a constant $c >0,$ 
such that for all  $t \leq cn,$ 
the integrality gap of $\opn{SA}^t(P)$    is $\Omega(n).$   
\end{theorem}

\begin{proof}
Observe that for every level of $\opn{SA}$ there is a suitable
projection of  $\mathcal{C}_{I}$ that yields  a legal core with
respect to  the
product variables used in that level. Therefore, the lower bound on the
size implied by Theorem
 \ref{theorem:mainCFL}  holds at all levels.  
The number of the inequalities of the $t$-level $\opn{SA}$ 
 relaxation  after the lifting and linearization stages, and before
 projection, 
 obtained from  any starting relaxation $P$ of size
$r$ is less than 
$r  \binom{n}{t} 2^{t}.$  By choosing  $t \leq cn,$   with $c$ sufficiently small, we obtain that 
$r  \binom{n}{t}  2^{t} \leq  r2^{\delta n}$ for a small $\delta > 0.$  By Theorem
 \ref{theorem:mainCFL} we get that for this value of $t,$  the integrality gap on the given
family of instances is $\Omega(n).$
This is asymptotically tight since SA is known to produce an exact formulation after $3n$ levels (the number of integer variables).
\end{proof}

We obtain as a  direct consequence a lower bound on the size of formulations 
that use only the classic variables $y_i,$ $x_{ij}.$ 

\begin{corollary}  \label{cor:natural}
Let $P$  be 
any linear  relaxation 
of the  \cfl\ polytope that uses the  encoding
$\mathcal{N}_{\text{\cfl}}$ and has integrality gap  $o(N),$  where $N$ is the number of facilities. Then $P$  has size 
$2^{\Omega(N)}.$ 
\end{corollary}

\opt{app}{
\subsection{Proof sketch for Lemma \ref{mainlemma}}
\label{section:keylemma-sketch}
}

\opt{full}{
\section{Proof of Lemma \ref{mainlemma}}\label{section:keylemma}
}

\opt{full}{
In the  first part of the proof  we will show that  by exchanging some
measure  of   some  components  of  the   two  product  vectors
$z_{k,l},z_{k,l'}$ of the core, we can construct two new
product  vectors   $z^*_{k,l},z^*_{k,l'}$  each  of   which  is
feasible for $\hat{D}.$ To establish this feasibility we will show for
each of them that it is a convex combination over 
vectors of the form $f(y,x)$ 
where $(y,x)$ are feasible mixed integer solutions in the
\cfl\ polytope $P$ and $f$ is the mixed  product section. 

Consider  the  two  sets  of  facilities $l-l'$  and  $l'-l$.  Clearly
$|l-l'|=|l'-l|>0$, since  $l\neq l'$ and $|l|=|l'|=n$.  We construct a
product vector  $z^*_{k,l}$ based on $z_{k,l}$  and making some
alterations and, symmetrically, a product vector $z^*_{k,l'}$ based
on $z_{k,l'}$.

}

\opt{app}{
In the  first part of the proof  we will show that  by exchanging some
measure  of   some  components  of  the   two  product  vectors
$z_{k,l},z_{k,l'}$ of the core, we can construct two new
product  vectors   $z^*_{k,l},z^*_{k,l'}$  each  of   which  is
feasible for $\hat{D}.$
Consider  the  two  sets  of  facilities $l-l'$  and  $l'-l$.  Clearly
$|l-l'|=|l'-l|>0$, since  $l\neq l'$ and $|l|=|l'|=n$.  We construct a
product vector  $z^*_{k,l}$ based on $z_{k,l}$  and making some
alterations and, symmetrically, a product vector $z^*_{k,l'}$ based
on $z_{k,l'}$. We give below the construction of $z^*_{k,l}$.
}

\centerline{\hrulefill
Construction of  $z^*_{k,l}$ \hrulefill}

For any set $\mathcal{E}$  
containing only facilities from $F-l'$ with at least one from $l-l'$: 
 $z^*_{k,l}(\mathcal{E})=z_{k,l}(\mathcal{E})+E_{\mathcal{D}_{k,l'}}[\mathcal{E}\cap
  case1]$
(Similarly, for  any $i,j$,
$z^*_{k,l}(\mathcal{E}x_{ij})=z_{k,l}(\mathcal{E}x_{ij})+E_{\mathcal{D}_{k,l'}}[\mathcal{E}x_{ij}
  \cap
  case1]$ ). 
In the case set $\mathcal{E}$ contains only facilities from
$F-l$ with at least one from $l'-l$ we have
$z^*_{k,l}(\mathcal{E})=z_{k,l}(\mathcal{E})-E_{\mathcal{D}_{k,l}}[\mathcal{E}\cap
  case1]$. 
(Similarly, for  any $i,j$,
$z^*_{k,l}(\mathcal{E}x_{ij})=z_{k,l}(\mathcal{E}x_{ij}) - 
E_{\mathcal{D}_{k,l}}[\mathcal{E}x_{ij}
  \cap
  case1]$ ). 
In any other case and for any $i,j$ 
let $z^*_{k,l}(\mathcal{E})=z_{k,l}(\mathcal{E})$ and
$z^*_{k,l}(\mathcal{E}x_{ij})=z_{k,l}(\mathcal{E}x_{ij}).$ \hrulefill

\opt{app}{

Next we show, and this is by far the most complicated part 
of the proof,  that the constructed $z^*_{k,l}$ and $z^*_{k,l'}$ are
indeed the  expected 
vectors of distributions $\mathcal{D}^*_{k,l}$ and
$\mathcal{D}^*_{k,l'},$ respectively, over feasible mixed integer
product solutions. 
In the last step of the proof  we show the following, which is an easy
consequence of the construction of $z^*_{k,l}$  and  $z^*_{k,l'}.$ 
\begin{claim}
 $1/2(z^*_{k,l}+z^*_{k,l'}) \in \operatorname{conv}(z_{k,l},z_{k,l'})$. 
\end{claim}
This last claim completes the 
 proof of Lemma \ref{mainlemma}. 
}

\opt{full}{
\centerline{\hrulefill
Construction of  $z^*_{k,l'}$ \hrulefill}

The construction of $z^*_{k,l'}$ is symmetric but we give the details for the sake of completeness. 
For any set $\mathcal{E}$  
containing only facilities from $F-l$ with at least one from $l'-l$: 
 $z^*_{k,l'}(\mathcal{E})=z_{k,l'}(\mathcal{E})+E_{\mathcal{D}_{k,l}}[\mathcal{E}\cap
  case1]$
(Similarly, for  any $i,j$,
$z^*_{k,l'}(\mathcal{E}x_{ij})=z_{k,l'}(\mathcal{E}x_{ij})+E_{\mathcal{D}_{k,l}}[\mathcal{E}x_{ij}
  \cap
  case1]$ ). 
In the case set $\mathcal{E}$ contains only facilities from
$F-l'$ with at least one from $l-l'$ we have
$z^*_{k,l'}(\mathcal{E})=z_{k,l'}(\mathcal{E})-E_{\mathcal{D}_{k,l'}}[\mathcal{E}\cap
  case1]$. 
(Similarly, for  any $i,j$, 
$z^*_{k,l'}(\mathcal{E}x_{ij})=z_{k,l'}(\mathcal{E}x_{ij})-E_{\mathcal{D}_{k,l'}}
[\mathcal{E}x_{ij}\cap
  case1]$.)
In any other case and for any $i,j$ 
let $z^*_{k,l'}(\mathcal{E})=z_{k,l'}(\mathcal{E})$ and
$z^*_{k,l'}(\mathcal{E}x_{ij})=z_{k,l'}(\mathcal{E}x_{ij}).$ \hrulefill

Next we show that the constructed $z^*_{k,l}$ and $z^*_{k,l'}$ are
indeed the  expected 
vectors of distributions $\mathcal{D}^*_{k,l}$ and
$\mathcal{D}^*_{k,l'},$ respectively, over feasible mixed integer product solutions. We will only give the proof for $z^*_{k,l}$ since the other case is similar. 

Before we  continue the proof, we first explain the  intuition behind
the   construction above.  Both   $z_{k,l}$  and  $z_{k,l'}$  are  not
derived from distributions over  feasible solutions  because in any
such feasible solution at least  one facility
from $l$ and $l'$ respectively   has to be  opened and
assigned  some demand.  By assigning  demand only  to the  set of
facilities in $k$ we cannot  satisfy the total demand without violating
the capacities. 
This is actually the main difference between the distributions
$\mathcal{D}^*_{k,l}$ and $\mathcal{D}_{k,l}.$ 
In Case $1^*$ there is at least
one  from $l$   opened but,  if we  were to  explain
$z_{k,l}$ and  $z_{k,l'}$ as resulting distributions over feasible
solutions,  
the total measure  of the
opening of  facilities in  $l$ and $l'$  respectively is too  small to
have some facility from any of those sets  opened 100\% of the time when Case
$1$ happens.  So in the construction  of $z^*_{k,l}$ we  will have 
all 
the facilities in $l-l'$ opened in Case $1^*$ and in the construction of
$z^*_{k,l'}$ we will have all  the facilities in $l'-l$ opened in Case
$1^*$. Since those  facilities are now opened a  greater fraction of the
time,  many  events  involving  them  will  have  their  probabilities
increased.  But  where  did  we  find  the  measure  to  increase  the
probability of those events? We construct  $z^*_{k,l}$ ``from'' $z_{k,l}$
by increasing  the probability of  those events in $z^*_{k,l}$  by the
same amount that we decrease  their probability in the construction of
$z^*_{k,l'}$ from $z_{k,l'}.$ Similarly we increase the probability
of those events in the construction of $z^*_{k,l'}$ from $z_{k,l'}$ by
the same amount that we decrease their probability in the construction
of $z^*_{k,l}$ from  $z_{k,l}$. If we prove the  validity of the above
and  moreover  prove  that  $\operatorname{conv}(z^*_{k,l},z^*_{k,l'})
\cap  \operatorname{conv}(z_{k,l},z_{k,l'}) \neq \varnothing$,  then 
Lemma~\ref{mainlemma} follows.

\begin{claim}\label{claim-for-lemma}
There is a distribution $\mathcal{D}^*_{k,l}$ over  mixed integer
product vectors which are feasible for $\hat{D}$ such that for
any $\mathcal{E}$ and any $i,$ $j$, $z^*_{k,l}(\mathcal{E})=E_{\mathcal{D}^*_{k,l}}[\mathcal{E}]$ and $z^*_{k,l}(\mathcal{E}x_{ij})=E_{\mathcal{D}^*_{k,l}}[\mathcal{E}x_{ij}]$.
\end{claim}

\begin{proof}
 Let $\mathcal{D}^*_{k,l}$ be the distribution defined by the
 following experiment. The vector $(\bar{y}, \bar{x})$ is the one
 defined in  Lemma~\ref{lemma::exp-vector}.\\

Facilities in $k$ are always opened.

\noi  {\it Case} $1^*$

 With probability $1-\frac{20}{n^2(1+1/n)}$ 
all facilities  in $F-l'$ are opened - all the facilities in $l'$ are closed. Evenly assign client demand to facilities in $k$ so each one takes exactly  $U$, and assign the remaining $2^{-n^2}$ demand evenly to the facilities in $l-l'$.

\noi  {\it Case} $2^*$

With  probability $\frac{20}{n^2(1+1/n)}$ pick at random a subset $q$ of the facilities in $F-k$ with at least one facility from $l$ and open them.

\noi  {\it Case} $2.a^*$

 If $q\neq F-k-l'$ assign randomly demand to facilities in $q\cap l$ so that each one of them takes $\frac{\sum_j \bar{x}_{ij}}{y_i}$ demand and the rest of the demand is equally distributed to the facilities in $k$. 

\noi  {\it Case} $2.b^*$
 
Otherwise,  when  $q= F-k-l'$, assign randomly a quantity $\frac{\sum_j \bar{x}_{ij}d}{\bar{y}_{i}}-2^{-n^2}\frac{P[\chi(case1^*)]}{P[\chi(case 2.b^*)]}$ to each of the facilities in $l-l'$ and assign the remaining demand evenly to the facilities in $k$.

 It is easy to see that  $\mathcal{D}^*_{k,l}$ is a distribution over feasible solutions for the instance 

\begin{proposition}\label{obvious}
Each outcome of the experiment defining $\mathcal{D}^*_{k,l}$ is a feasible solution to the instance.
\end{proposition}

\begin{proof}
In every outcome of the probabilistic experiment which induces the distribution $\mathcal{D}^{*}_{k,l}$ all the client demand is assigned to the opened facilities. It remains to show that the capacities $U$ are respected. Consider case $1$: the capacities of the facilities in $k$ are respected by construction (recall that in this case those facilities are saturated) and facilities in $l-l'$ share a total demand of $2^{-n^2} < U$. Now consider case $2$: a facility $i \in q$ is assigned at most $\frac{\sum_j \bar{x}_{ij}}{10/n^2}<U$ demand and since the rest of the demand is equally distributed among the facilities in $k$, each facility in $k$ takes at most (equality when $q=F-k-l'$)  
$\frac{(n^4+1)-[(|l-l'|)\frac{\sum_j \bar{x}_{ij}}{\bar{y}_{i}}-(2^{-n^2})\frac{E[\chi(case1^*)]}{E[\chi(case2.b^*)]}]}{|k|} < U$.
\end{proof}

To get Claim \ref{claim-for-lemma}, we prove the following proposition.

\begin{proposition}\label{prop-for-lemma}
For every set $\mathcal{E}$ and any $i,$ $j,$ we have that 
$z^*_{k,l}(\mathcal{E})=E_{\mathcal{D}^*_{k,l}}[\mathcal{E}]$ and 
$z^*_{k,l}(\mathcal{E}x_{ij})=E_{\mathcal{D}^*_{k,l}}[\mathcal{E}x_{ij}]$. 
\end{proposition}

\begin{proof}
Consider the following cases.

Case A.1

Assume that 
 set $\mathcal{E}$ contains facilities only from $F-l'$ with at least one from $l-l'$.
We have that $z^*_{k,l}(\mathcal{E})=z_{k,l}(\mathcal{E})+E_{\mathcal{D}_{k,l'}}[\mathcal{E}\cap case1]$. By the definition of $z_{k,l}(\mathcal{E})$ we have $z_{k,l}(\mathcal{E})=E_{\mathcal{D}_{k,l}}[\mathcal{E}]=E_{\mathcal{D}_{k,l}}[\mathcal{E}\cap case2]$ since in case $1$ of $\mathcal{D}_{k,l}$ none of the facilities in $l$ are opened. We also have that $E_{\mathcal{D}_{k,l}}[\mathcal{E}\cap case2]=E_{\mathcal{D}^*_{k,l}}[\mathcal{E}\cap case2^*]$
 since no assignment variable  appears in $\mathcal{E}$  and all the
 other elements of the experiments of cases $2$ and $2^*$ 
 induce the exact same distribution. We also have that $E_{\mathcal{D}_{k,l'}}[\mathcal{E}\cap case1]=E_{\mathcal{D}^*_{k,l}}[\mathcal{E}\cap case1^*]$ by the fact that when case $1$ happens in $\mathcal{D}_{z_{k,l'}}$ the facilities in $l-l'$ are opened 100\% and the same happens in $\mathcal{D}^*_{k,l}$, while the $2$ distributions agree on everything except the assignments in that case  by construction (recall again that no assignment variable appears in $\mathcal{E}$). So we have $z^*_{k,l}(\mathcal{E})=z_{k,l}(\mathcal{E})+E_{\mathcal{D}_{k,l'}}[\mathcal{E}\cap case1]=E_{\mathcal{D}^*_{k,l}}[\mathcal{E}\cap case2^*]+E_{\mathcal{D}^*_{k,l}}[\mathcal{E}\cap case1^*]=E_{\mathcal{D}^*_{k,l}}[\mathcal{E}]$ since cases $1^*$ and $2^*$ partition the probability space.

Case A.2.a

       Consider the case where $\mathcal{E}$ contains facilities only
       from $F-l'$ with at least one from $l-l'$ and let $x_{ij}$ be
       an assignment to some  facility $i \in k.$  We have that \\

$z^*_{k,l}(\mathcal{E}x_{ij})=z_{k,l}(\mathcal{E}x_{ij})+E_{\mathcal{D}_{k,l'}}[\mathcal{E}x_{ij}\cap case1]=$\\
$E_{\mathcal{D}_{k,l}}[\mathcal{E}x_{ij}\cap case2]+E_{\mathcal{D}_{k,l'}}[\mathcal{E}x_{ij}\cap case1]$(none of the facilities in $l$ are opened in case $1$ of $\mathcal{D}_{k,l}$).\\

Note that in case $1$ of $\mathcal{D}_{k,l}$ all the client demand is assigned to the facilities in $k$
while in case $1$ of $\mathcal{D}^*_{k,l}$ the demand assigned to facilities in $k$ is the total demand minus $2^{-n^2}$, which is assigned to the facilities in $l-l'$. Also note that assignments in $k$ are done evenly for all clients. Thus the last equation yields:\\

$(E_{\mathcal{D}^*_{k,l}}[\mathcal{E}x_{ij}\cap case2^*]-p)+(E_{\mathcal{D}^*_{k,l'}}[\mathcal{E}x_{ij}\cap case1^*]+p)=$(here $p=2^{-n^2}E_{\mathcal{D}_{k,l}}[\chi(case1)]/|k||C|$)\\
$E_{\mathcal{D}^*_{k,l}}[\mathcal{E}x_{ij}].$\\

Case A.2.b

Now, if set $\mathcal{E}$ contains facilities only from $F-l'$ and $x_{ij}$ is an assignment variable of $l-l'$ then, again,  $z^*_{k,l}(\mathcal{E}x_{ij})=z_{k,l}(\mathcal{E}x_{ij})+E_{\mathcal{D}_{k,l'}}[\mathcal{E}x_{ij}\cap case1]$. But now $E_{\mathcal{D}_{k,l'}}[\mathcal{E}x_{ij}\cap case1]=0$ since none of $l-l'$ are assigned any demand when case $1$ happens in $\mathcal{D}_{k,l'}$. From the definition of the core we have $z_{k,l}(\mathcal{E}x_{ij})=E_{\mathcal{D}_{k,l}}[\mathcal{E}x_{ij}]=E_{\mathcal{D}_{k,l}}[\mathcal{E}x_{ij}\cap case2]$ (in case 1 all assignments to $l$ are zero).  We have:



$z_{k,l}(\mathcal{E}x_{ij})=E_{\mathcal{D}_{k,l}}[\mathcal{E}x_{ij}\cap case2].$\\

Note that in case $2.b*$ of $\mathcal{D}^*_{k,l}$ the total demand assigned to $i$ is $2^{-n^2}\frac{P[\chi(case1^*)]}{P[\chi(case 2.b^*)]}$ less than the total demand assigned to it in the corresponding wrt to $q$ case of $\mathcal{D}_{k,l}$. Thus, again by symmetry of the assignments, the last expression is equal to:\\

$E_{\mathcal{D}^*_{k,l}}[\mathcal{E}x_{ij}\cap case2^*]+r=$(here $r=2^{-n^2} E_{\mathcal{D}_{k,l}}[\chi(case1)]/|l-l'||C|$)\\
$E_{\mathcal{D}^*_{k,l}}[\mathcal{E}x_{ij}].$\\
So once again $z^*_{k,l}(\mathcal{E}x_{ij})=E_{\mathcal{D}^*_{k,l}}[\mathcal{E}x_{ij}]$. 

Case B

Consider the case where $\mathcal{E}$ contains facilities in $l'-l$ and so  $z^*_{k,l}(\mathcal{E})=z_{k,l}(\mathcal{E})-E_{\mathcal{D}_{k,l}}[\mathcal{E}\cap case1]$. By definition of the core $z_{k,l}(\mathcal{E})=E_{\mathcal{D}_{k,l}}[\mathcal{E}]$. So
$z^*_{k,l}(\mathcal{E})=E_{\mathcal{D}_{k,l}}[\mathcal{E}\cap case2]$. On the other hand,  $E_{\mathcal{D}^*_{k,l}}[\mathcal{E}]=E_{\mathcal{D}^*_{k,l}}[\mathcal{E}\cap case1^*]+E_{\mathcal{D}^*_{k,l}}[\mathcal{E}\cap case2^*]=E_{\mathcal{D}^*_{k,l}}[\mathcal{E}\cap case2^*]$ since $E_{\mathcal{D}^*_{k,l}}[\mathcal{E}\cap case1^*]=0$ (the facilities in $l'-l$ are always closed in case $1$ of $\mathcal{D}^*_{k,l}$). Case $2$ of $\mathcal{D}_{k,l}$ and case $2^*$ of $\mathcal{D}^*_{k,l}$ differ only on the case where $q=F-k-l'$. Since $\mathcal{E}$ contains facilities in $l'-l$ it cannot be the case $q=F-k-l'$. 
So we have that $E_{\mathcal{D}_{k,l}}[\mathcal{E}\cap case2]=E_{\mathcal{D}^*_{k,l}}[\mathcal{E}\cap case2^*]$. So once again we have  $z^*_{k,l}(\mathcal{E})=E_{\mathcal{D}^*_{k,l}}[\mathcal{E}]$. The exact same arguments are valid in case we consider $\mathcal{E}x_{ij}$ for any assignment variable $x_{ij}$.

Case C

For every other set $\mathcal{E}$ and any $(i,j)$ we  have $z^*_{k,l}(\mathcal{E})=z_{k,l}(\mathcal{E})=E_{\mathcal{D}_{k,l}}[\mathcal{E}]$($z^*_{k,l}(\mathcal{E}x_{ij})=z_{k,l}(\mathcal{E}x_{ij})=E_{\mathcal{D}_{k,l}}[\mathcal{E}x_{ij}]$) which is equal to $E_{\mathcal{D}^*_{k,l}}[\mathcal{E}]$($E_{\mathcal{D}^*_{k,l}}[\mathcal{E}x_{ij}]$)  by construction of the distributions $\mathcal{D}^*_{k,l}$ and $\mathcal{D}^*_{k,l}$.

The proof of Proposition \ref{prop-for-lemma} is complete.
\end{proof}

To complete the proof of the claim note that $\mathcal{D}^*_{k,l}$ can also be seen as a distribution over product vectors by associating each mixed integer outcome $(y,x)$ of the experiment with the product vector $f(y,x)$ -- recall that $f(y,x)$ is the mixed product section.  Observe that the expectations $E_{\mathcal{D}^*_{k,l}}[\mathcal{E}]$ and $E_{\mathcal{D}^*_{k,l}}[\mathcal{E}x_{ij}]$ are exactly the expectations of the corresponding components of the product vectors $f(y,x)$.

The proof of Claim \ref{claim-for-lemma} is complete.
\end{proof}

Finally, we show the following

\begin{claim}
 $1/2(z^*_{k,l}+z^*_{k,l'}) \in \operatorname{conv}(z_{k,l},z_{k,l'})$. 
\end{claim}

\begin{proof}
We shall actually show that $1/2(z^*_{k,l}+z^*_{k,l'}) = 1/2(z_{k,l}+z_{k,l'})$. 
Let  $\mathcal{E}$ be a  set containing facilities only from $F-l'$ with at least one from $l-l'$.
Then 
$z^*_{k,l}(\mathcal{E})+z^*_{k,l'}(\mathcal{E})=z_{k,l}(\mathcal{E})+ E_{\mathcal{D}_{k,l'}}[\mathcal{E}\cap case1] + z_{k,l'}(\mathcal{E})- E_{\mathcal{D}_{k,l'}}[\mathcal{E}\cap case1]=z_{k,l}(\mathcal{E})+z_{k,l'}(\mathcal{E})$.

Let $\mathcal{E}$ be a  set containing facilities only from $F-l$ with at least one from $l'-l$. We have 
$z^*_{k,l}(\mathcal{E})+z^*_{k,l'}(\mathcal{E})=z_{k,l}(\mathcal{E})- E_{\mathcal{D}_{k,l}}[\mathcal{E}\cap case1] + z_{k,l'}(\mathcal{E})+ E_{\mathcal{D}_{k,l}}[\mathcal{E}\cap case1]=z_{k,l}(\mathcal{E})+z_{k,l'}(\mathcal{E})$.


In the remaining cases for the set $\mathcal{E}$ we simply have $z^*_{k,l}(\mathcal{E})+z^*_{k,l'}(\mathcal{E})=z_{k,l}(\mathcal{E})+z_{k,l'}(\mathcal{E})$
by construction. The exact same arguments are valid in case we consider $\mathcal{E}x_{ij}$ for any assignment variable $x_{ij}$.
\end{proof}

The proof of Lemma \ref{mainlemma} is complete.

}      

\section{Discussion}

In the proof of our result for \cfl\ we provided a core whose underlying hypergraph
is actually a simple graph and moreover a clique. For other problems,  especially for  \zeone  polytopes, 
we believe that the power of general hypergraphs needs to
be exploited,  if one wishes to derive a tight bound on 
the extension complexity. 
Observe that our methodology requires  only the existence of a
suitable core, and thus,  one could possibly employ
probabilistic arguments to prove the existence of suitable hypergraphs of
high chromatic number.

In the case of mixed integer polytopes, we believe that the mixed 
product relaxations can be shown to be strong enough to simulate any
extended formulation, as is the case for  product
formulations and \zeone polytopes.

\subparagraph*{Acknowledgements}
We thank the anonymous reviewers of an earlier version for  valuable  comments.

\bibliography{../../bibliography-ver2}

\end{document}